\documentclass[a4paper,12pt,reqno]{elsarticle}
    \makeatletter
    \def\ps@pprintTitle{%
      \let\@oddhead\@empty
      \let\@evenhead\@empty
      \def\@oddfoot{\reset@font\hfil\thepage\hfil}
      \let\@evenfoot\@oddfoot
    }
\usepackage{amsmath}%
\usepackage{amsfonts}%
\usepackage{amssymb}%
\usepackage{graphicx}
\usepackage{color}
\usepackage{setspace}
\usepackage[margin=31.4mm]{geometry}
\newtheorem{theorem}{Theorem}

\newtheorem{corollary}{Corollary}
\newtheorem{criterion}{Criterion}

\newtheorem{example}{Example}

\numberwithin{equation}{section}

\newproof{proof}{Proof}

\newcommand{\E}{\mathop{\mathbb E}}
\newcommand{\p}{\mathop{\mathbf P}}

\begin{document}

\title{Market Share Indicates Quality}
\author{Amir Ban}
\ead{amirban@netvision.net.il}
\address{Center for the Study of Rationality, Hebrew University, Jerusalem, Israel}
\author{Nati Linial}
\ead{nati@cs.huji.ac.il} 
\address{School of Computer Science and Engineering, Hebrew University, Jerusalem, Israel}
\date{}


\begin{abstract}

Market share and quality, or customer satisfaction, go together. Yet inferring one from the other appears difficult. Indeed, such an inference would need detailed information about customer behavior, and might be clouded by modes of behavior such as herding (following popularity) or elitism, where customers avoid popular products. We investigate a fixed-price model where customers are informed about their history with products and about market share data. We find that it is in fact correct to make a Bayesian inference that the product with the higher market share has the better quality under few and unrestrictive assumptions on customer behavior.

\end{abstract}

\maketitle

\section{Introduction}

Common wisdom holds that a full restaurant is a good one, or certainly better than its empty neighbor. The purpose of this paper is to discover some minimal assumptions on the rationality of customers under which this folk wisdom can be mathematically justified. For example, this conclusion certainly does not hold in places where the (admittedly strange) general preference is for food of poor taste.  We formulate a simple model of a market in which customers have several products available to them. Each product has an innate unknown {\em quality} which is the probability that a customer who consumes it is {\em satisfied}. We find very mild sufficient conditions under which a larger market share indicates higher product quality.

Intuitively, quality goes hand in hand with market share, and indeed a manufacturer's pursuit of quality is usually rationalized as a way of maximizing economic benefit. While not conclusive, empirical studies (e.g. Anderson et. al. \cite{Anderson}, Rust and Zahorik \cite{Rust}) generally support a positive correlation between quality and market share.

Market share affects real or perceived quality through various mechanisms, including the creation of network externalities. Yet the direction of this potential impact is ambiguous in the literature. An empirical study by Hellofs and Jacobson\cite{Hellofs} concludes ``The view that \ldots market share and product quality are reinforcing mechanisms \ldots seems premature''. Our study makes room for both positive and negative impact, through the modeling of ``herders'' and ``elitists''.

We model quality as a probability for customer satisfaction. We consider quality to be a hidden, constant attribute of a product, which may be inferred, but not directly observed or learned from an authority. This is a widespread scenario that customers realistically face. In the exceptions, e.g. when offered a Ferrari and a Fiat, price or other differentiation will typically also exist. The extent to which customers may remain unaware of quality has been studied by, e.g., Mitra and Golder\cite{Mitra}, who showed that only a fraction of actual quality change registers with customers, and then typically with a delay of years.

Caminal and Vives\cite{Caminal} studied the effect of firms' signaling of quality on imperfectly informed customers. The signaling is facilitated by strategic price reductions, which increase perceptions of quality by increasing market share. As customers best-respond with their own strategy, in one possible equilibrium market share indicates quality. Our motivation is different. We seek to illuminate the inter-relationship of market share and quality under the least possible restrictions on customer behavior, and independently of signaling, for several reasons: (i) A large body of economic research indicates that customers' decisions use {\em bounded} rationality. (ii) Strategic signaling of quality is absent when there is no interest or possibility of doing so, e.g. when prices are set by a regulator (health insurance), or by a non-profit (museum entry tickets), or by a channel (movie tickets), or due to the shortness of the interaction (shore restaurants serving a boatload of tourists stopping for lunch will inevitably have same prices). (iii) The sparseness of theoretical results, and the inconclusiveness of empirical results, increases the value of a basic result. Especially so since, as we will point out, and as the extensive literature on herding shows, there is nothing simple or inevitable about it.

We ascribe to each customer a strategy of whether to consume each of the products, guided by market share and by one's own past personal history. When described as a behavioral strategy, i.e. by probabilities for consuming the product for each of the customer's information sets, we call it the customer's {\em partiality} strategy for that product. No connection between consumption of different products is assumed. E.g., a customer may consume all products simultaneously or none. We make no assumption that customers have uniform strategies, or that any customer's strategy is optimal.

We make two mild assumptions on customers' strategies, that customers (i) do not prefer negative over positive experiences with products, and (ii) do not prefer products with low market share over those with a higher one. Additionally, we assume an undifferentiated market, where products are {\em a priori} equal in the eyes of customers. Under these assumptions we can show the validity of inferences from market share to quality.

The current work also sets the ground for future work that will include variable prices. A close inspection of our methodology reveals  its suitability for such a task. It is particularly encouraging to see its success in reaching general results from a few weak assumptions.  One of our results, that the probability for market leadership increases with quality (Theorem \ref{theorem2}) applies to variable-price markets. In the Discussion we suggest a path to extending the main result for variable prices.

Though our research is not framed in the canonical terms of action payoffs and Bayesian updating/learning, it in fact conforms to it: A preference for satisfaction over dissatisfaction is all we need to assume about payoffs, and, the monotonicity property, which is introduced in the next section, has a property common to all Bayesian updating rules: A strictly superior sequence of events results in a superior posterior.

\subsection{Strategy Dependence on Personal History}

When customers (illogically) prefer dissatisfaction over satisfaction, market share clearly does not indicate quality. There are multiple ways to exclude or limit this from our framework. For example, we could require that customers' strategies be consistent with their {\em average} satisfaction with products. But this is already restrictive: It is, for example, not unreasonable to prefer a product used satisfactorily 18 times out of 20 trials over a product used just once satisfactorily; Or to give more weight to more recent trials.

Therefore we adopt a tamer restriction, which we call {\em monotonicity}: Namely, that customers recall the outcomes of their experiences with products, in the order that they happened, and if that history is definitely superior, on an experience-by-experience basis, their partiality strategy\begin{samepage}\footnote{Here and hereafter we use ``partiality strategy'' as shorthand for the probability of choosing the action of consumption under that strategy.}\end{samepage} to the product will be equal or higher. As an example, if a customer has a {\em fail-success-fail} history with a product (on the 3 occasions she elected to use it), then her partiality strategy after such a history will not be higher than if her history would have been {\em fail-success-success}, as the latter is superior by having a success where the former has failure, and is otherwise the same. No restriction is made on strategy after, e.g., the history {\em success-fail-success} (incomparable on an experience-by-experience basis), or {\em success-success} (incomparable due to a different number of experiences). Nor does it restrict the customer's strategy to other products, or other customers' strategies, as each can be formed independently within our framework.

This is possibly the lightest restriction on customer strategies we could make that conforms to common sense. When customers' strategies are guided solely by their history, we show it is sufficient to establish our result. Even for this restricted scenario, the conclusion is deeper than suggested by the assertion's simplicity: Attempted proofs must deal with a side result to which we allude in the Discussion:  When quality varies with time, inferences from quality to market share or vice versa are, as a rule, invalid.

\subsection{Strategy Dependence on Market Share}

Customers may base their strategy on market share itself. This may take several forms:  Customers may be fully or partially informed of market share, by, e.g. knowing product sales figures, or the ranking of the top-selling products. Smallwood and Conlisk\cite{Smallwood} considered a market that evolves based on products having an intrinsic probability for breakdown and customers switching products randomly weighed by a function of market share. The present authors \cite{Ban} considered a system where customers are influenced by history and reputation, where ``reputation'' under a suitable choice of model parameters represents market share. Word-of-mouth, i.e. asking or following others, is in effect a sampling of market share. Ellison and Fudenberg\cite{Ellison} considered a model of learning involving both personal history and word-of-mouth communication in which technologies perform stochastically based on an underlying quality parameter. Information cascades, starting with Bikhchandani et. al.\cite{Bikhchandani}, consider the inferences that observers can make on the quality of a service based on the customers queuing for that service and how informed those customers are known to be. They show this leads to {\em herding}, the phenomenon where customers accumulate due to the presence of others. Smallwood and Conlisk\cite{Smallwood} as well as Ban and Linial\cite{Ban} also show that lower quality products can maintain higher market share indefinitely. Lest the reader suspect that the main result is self-evident when customers are Bayesian, since ``better quality always wins with Bayesian learning", note that this assumption is refuted by the above cited papers which are framed in strictly Bayesian terms. This false intuition is dispelled as well by the $k$-armed bandit problem, where an unlucky start may cause a Bayesian player to miss his best option for a long time or even indefinitely.

Another potential feature of customer behavior we call {\em elitism}: Customers who intentionally avoid the most popular products. This may be due to a wish to differentiate oneself from the crowd, or to a belief that popular choices are second-rate, or any other reason.

Herding and elitism seem to cast doubt on our thesis. Herding, in particular, seems to pull the rug from underneath our sought conclusion. For example, if more than half of the customers at any point in time consume the market-leading product, and it alone, then market leadership is self-perpetuating regardless of product qualities and regardless of how other customers behave. No monotonicity assumptions are violated (for added credibility, assume leader-following customers are one-shot with no experiences to rely on), yet market share indicates nothing regarding product qualities.  However, as we demonstrate, herding poses no problem to a market-to-quality inference: While a lower-quality product may sometimes prevail in market share, this will always have a lower probability than the alternative, and market share data {\em per se} is of no help in recognizing that such an anomaly is occurring. Defining a customer to be {\em weakly herding} if greater market share makes her more likely to consume a product, or has no effect on her behavior, we demonstrate that when all customers are weakly herding (in addition to being monotone on their product histories), market share is a valid signal for quality.

As for elitism, we believe, but do not analyze in the current paper, that if  outweighed (in some sense) by herding, our thesis is still valid. Markets in which customer elitism is dominant turn out to be chaotic and difficult to analyze. In the Discussion we give an example where such a market does not adhere to our thesis. However, such markets seem far-fetched and so of low economic significance.

\subsection{No Other Differentiation}

Our result applies to undifferentiated markets, where all products are {\em a priori} equal in the eyes of customers. When customers distinguish between products by price, brand name, etc., or in captive markets, market share may be a reflection of the existing differentiation rather than of quality.

In our model, this translates to a requirement of {\em anonymity} of products in customer strategies, meaning that customers' strategies are invariant under a change of product labels.

\subsection{Organization of this paper}

The rest of this paper is organized as follows: Section \ref{themodel} and \ref{history} analyze markets where customers are guided entirely by their product history, with section \ref{themodel} devoted to describing the model and section \ref{history} stating and proving our proposition in such markets. Subsequently we analyze markets where customers are aware of product market share and take it into account, with section \ref{themodelplus} devoted to refining the model for such markets, while section \ref{marketshare} states and proves our proposition. Formally speaking the results in Sections \ref{themodelplus} and \ref{marketshare} subsume those of Sections \ref{themodel} and \ref{history}, but we feel that this organization of the material makes it easier for the reader to follow. Conclusions are given in section \ref{conclusion}.

\section*{Acknowledgement}

The original version of this paper was based on a more complicated argument. The proof method that we finally adopted was generously offered by an anonymous referee to whom we are very grateful.

\section{Basic Model, When Only History Matters}
\label{themodel}

In our model, customers make decisions regarding products in {\em rounds} of discrete time $t = 1,2,\ldots$. At each round, a customer has an action set $\{C,N\}$, where $C:=$ consume the product, $N := $ do not consume the product.  If she consumes the product, she will, with probability given by the product's quality $q \in [0,1]$, be satisfied, in which case the round is called an {\em $S$-round}, or else dissatisfied, in which case the round is called an {\em $F$-round}. If she chooses not to consume the product, the round is called an {\em $N$-round}. 

A customer's $t$-{\em deep history} with a product is a member of $\mathcal{H}_t := \{S, N, F\}^t$. The set of all histories are denoted $\mathcal{H} := \bigcup_{k=0}^\infty \mathcal{H}_k$. For $Z \in \mathcal{H}_t$ we mark $Z$'s {\em depth} $|Z| := t$. $Z$ is composed of {\em events}, $Z = (Z(1), \ldots, Z(t))$ with $Z(k)$ standing for the event in round $k$. $Z^k$ stands for the sequence of events until round $k$, i.e. $Z^k = (Z_1, \ldots, Z_k) \in \mathcal{H}_k$, so that $Z = Z^t$. 

The customer's {\em partiality strategy}, $\sigma: \mathcal{H} \to [0,1]$, is her behavioral strategy given her information set, which in this basic model is her history with the product at the time of decision. The partiality strategy is completely specified by specifying $\sigma(Z)$, the probability for action $C$, for each history $Z \in \mathcal{H}$.

For $V \in \{S,N,F\}$, $ZV$ stands for $(Z(1), \ldots , Z(t), V) \in \mathcal{H}_{t+1}$, with $Z(k)$ standing for the event in round $k$.

We further define $S(Z)$ (resp. $F(Z)$, $N(Z)$) as the number of $S$-rounds (resp. $F, N$-rounds) in $Z$, i.e., the number of indices $i$ for which $Z(i)=S$ (resp. $F, N$). The {\em consumption} of $Z$ is defined as $con(Z) := S(Z)+F(Z)$. The {\em digest} of $Z$, denoted $dig(Z) \in \mathcal{H}_{con(Z)}$, is defined as the history that we obtain when we omit all the $N$-rounds from $Z$ while maintaining the order of the remaining rounds.

Let $Z_1, Z_2 \in \mathcal{H}_t$, with $con(Z_1)=con(Z_2)$, and let $D_1 := dig(Z_1), D_2 := dig(Z_2)$. We say that $Z_1$ is {\em superior} to $Z_2$, denoted $Z_1 \succeq Z_2$ if there is no index $i$ for which $D_1(i)=F$ and $D_2(i)=S$.

A partiality strategy $\sigma(\cdot)$ is called {\em monotone} if 
\begin{equation}\label{define_monotone}
\sigma(Z_1) \geq \sigma(Z_2) \text{~whenever~} Z_1 \succeq Z_2.\footnote{In more general models where the customer's strategy is dependent on more information, e.g. history of other products, or market shares, the generalization of this will be ``all else being equal''. E.g. $\sigma(Z_1,X,Y,\ldots) \geq \sigma(Z_2,X,Y,\ldots) \text{~whenever~} Z_1 \succeq Z_2$, where $X,Y, \ldots$ denote any other information-set variables. See e.g. the definition of monotonicity in Section \ref{themodelplus}.}
\end{equation}

\section{The Main Theorem When Only Product History Matters}
\label{history}

In the current section we focus on the situation of a monotone partiality strategy that depends only on history. What can be said about the probability that the consumption up to time $t$, is $\ge x$ for arbitrary $t$ and $x$? As the following theorem shows, this probability is a non-decreasing function of the product quality $q$.

\begin{theorem}
\label{theorem1}
Fix a monotone partiality strategy $\sigma$, and nonnegative integers $t, x$. Then
\begin{equation}
\label{event1}
\frac{d}{dq} \p \Bigl[con(Z) \geq x | Z \in \mathcal{H}_t \Bigr] \geq 0
\end{equation}
where the probability space is  $\mathcal{H}$.
\end{theorem}

\begin{proof}

We define a Markov chain on histories, i.e. a Markov chain with state space $\mathcal{H}$ that describes the possible transitions between histories and their probabilities. All transitions are from a member $Z \in \mathcal{H}_t$ to an extension $Z' \in \mathcal{H}_{t+1}$ with the following probabilities 
\begin{equation}
\label{markov}
       \begin{array}{ll}
           Z \to ZS & $with probability $ q\cdot \sigma(Z) \\
           Z \to ZN & $with probability $ 1 - \sigma(Z) \\
           Z \to ZF & $with probability $ (1-q)\cdot \sigma(Z)
     \end{array}
\end{equation}

Consider the probability of reaching $Z \in \mathcal{H}$ as we start from the empty history and move along the Markov chain. It is convenient to express this probability as $$\p[Z;q] := c(Z) Q(Z; q)$$ $$Q(Z; q) := q^{S(Z)} (1-q)^{F(Z)}.$$ We refer to $c(Z)$ as the {\em ex-ante} function corresponding to strategy $\sigma(\cdot)$. Following from \eqref{markov}, its value is recursively defined by:
\begin{equation}
c(Z^k) = \left\{
\label{ex-ante}
       \begin{array}{ll}
           1 & k = 0 \\
           \sigma(Z^{k-1}) c(Z^{k-1})  & Z(k) \neq N \\
           \bigl[1-\sigma(Z^{k-1})\bigr]c(Z^{k-1}) & Z(k) = N
     \end{array}
    \right.
\end{equation}

For example $c(FNSSN) = \sigma(\emptyset)[1-\sigma(F)]\sigma(FN)\sigma(FNS)[1-\sigma(FNSS)]$ where $\emptyset$ denotes the empty history. Observe that $c(Z)$ is a product of $|Z|$ factors. The factor has the form $\sigma(\cdot)$ for each consumption event, and $1 - \sigma(\cdot)$ where the history has an $N$-event. The arguments of $\sigma$ in the factors run over all $|Z|$ tails of $Z$.

Let $q', q \in [0,1]$ s.t. $q' > q$. For the proof, we will construct a joint distribution of two types of history:
\begin{enumerate}
\item The distribution of histories $\in \mathcal{H}_t$ under product quality $q$. Histories from this distribution are denoted by $Z$.
\item The distribution of histories $\in \mathcal{H}_t$ under product quality $q'$.
Histories from this distribution are denoted by $Z'$.
\end{enumerate}

The joint distribution is defined for the purpose of our analysis. It should be considered a proof technique and not a representation of any real customer behavior.

The construction assigns to each pair of histories $Z, Z' \in \mathcal{H}_t$ a joint probability $f(Z,Z') \in [0,1]$. Assuming $f(Z,Z')$ to be tabulated in a table whose rows correspond to values of $Z$ and columns correspond to values of $Z'$, a valid joint distribution must have rows summing to the correct history event probabilities in both rows and columns, i.e. for every $Z \in \mathcal{H}_t$:
\begin{equation}
\label{row_marginal}
\sum\limits_{Z' \in \mathcal{H}_t} f(Z,Z') = \p[Z;q] = c(Z) Q(Z; q)
\end{equation}

\noindent and for every $Z' \in \mathcal{H}_t$:
\begin{equation}
\label{column_marginal}
\sum\limits_{Z \in \mathcal{H}_t} f(Z,Z') = \p[Z';q'] = c(Z') Q(Z'; q')
\end{equation}

Our construction restricts the joint distribution's support (i.e. those pairs $Z, Z'$ where $f(Z,Z') > 0$) by the following criterion:

\begin{criterion}
\label{support}
For every $0 < k \leq t$, there exists $0 < l \leq k$ s.t. $dig(Z'^l) \succeq dig(Z^k)$.

Consequently, for every $0 < k \leq t$, $con(Z'^k) \geq con(Z^k)$, and $con(Z'^k) = con(Z^k) \Rightarrow Z'^k \succeq Z^k \Rightarrow \sigma(Z'^k) \geq \sigma(Z^k)$.
\end{criterion}

We specify the joint distribution in two different ways: In the first, it is specified how the row marginal probabilities, i.e. probabilities of $Z$ events, are split into the row's individual probabilities. The construction will explicitly guarantee \eqref{row_marginal} and Criterion \ref{support}.  In the second, it is specified how the column marginal probabilities, i.e. probabilities of $Z'$ events, are split into the column's individual probabilities. The construction will explicitly guarantee \eqref{column_marginal} and Criterion \ref{support}. Finally, we will demonstrate that the two constructions coincide and yield the same probabilities, thus proving that it is in fact a joint probability fulfilling \eqref{row_marginal}, \eqref{column_marginal} and Criterion \ref{support}.

The first, row construction, is specified by:

Let $Z, Z', X, Y \in \mathcal{H}_t$ (we use $X$ in place of $Z$, and $Y$ in place of $Z'$ for better readability), and $0 < k \leq t$. We define below probability functions $g_1(X,Y,k), h_1(X,Y,k)$. Based on these, the probabilities $f(Z,Z')$ will be given by:
\begin{align}
\label{g_def}
& g(X,Y,k) :=  \left[
       \begin{array}{ll}
	 g_1(X,Y,k)	& Y(k) \neq N \\
	 1 - g_1(X,Y,k)	& Y(k) = N
     \end{array}
    \right] \\
\label{h_def}
& h(X,Y,k) =  \left[
       \begin{array}{ll}
	 h_1(X,Y,k)	& Y(k) = S \\
	 1 - h_1(X,Y,k)	& Y(k) = F \\
	1	& Y(k) = N
     \end{array}
    \right] \\
\label{r_def}
& r(X,Y,k) := g(X,Y,k) h(X,Y,k) \\
\label{R_def}
& R(Z,Z') := \prod\limits_{k=1}^t r(Z,Z',k) \\
\label{f_def}
& f(Z,Z') := R(Z,Z') \p[Z;q]
\end{align}

From Equations \eqref{g_def} to \eqref{r_def} we derive:
\begin{align}
r(X,Y,k) =  \left[
       \begin{array}{ll}
	 g(X,Y,k) h(X,Y,k)	& Y(k) = S \\
	 g(X,Y,k) [1 - h(X,Y,k)]	& Y(k) = F \\
	1-g(X,Y,k)	& Y(k) = N
     \end{array}
    \right]
\end{align}

\noindent The three values of which sum to $1$. It follows that for every given values of $Z$ and $Z'^{k-1}$ there holds $\sum r(Z,Z',k)=1$ where the sum is over  $Z'(k)=N, S$ and $F$. By \eqref{R_def} it follows that:
\begin{equation}
\sum\limits_{Z' \in \mathcal{H}_t} R(Z,Z') = 1
\end{equation}

\noindent from which \eqref{row_marginal} follows. Note that this holds regardless of the choice of $g_1, h_1$, which is made so as to guarantee that Criterion~\ref{support} holds. The definition of $g_1$ and $h_1$ follows:

Mark $c := con(X^{k-1})$, $c' := con(Y^{k-1})$:

\begin{equation}
\label{g1_def}
g_1(X,Y,k) =  \left[
       \begin{array}{ll}
\left[
       \begin{array}{ll}
	1 & X(k) \neq N \\
	\frac{\sigma(Y^{k-1}) - \sigma(X^{k-1})}{1 - \sigma(X^{k-1})} & X(k) = N \\
     \end{array}
    \right] & c = c' \\
	\sigma(Y^{k-1}) & c \neq c' \\
     \end{array}
    \right]
\end{equation}

\begin{equation}
\label{h1_def}
h_1(X,Y,k) =  \left[
       \begin{array}{ll}
\left[
       \begin{array}{ll}
	1 & [dig(X)](c+1) = S \\
	\frac{q'-q}{1-q}  & [dig(X)](c+1) = F \\
     \end{array}
    \right] & c < con(X) \\
	q'  & c \geq con(X) \\
     \end{array}
    \right]
\end{equation}

An examination of \eqref{g_def} to \eqref{h1_def} leads to the following observations:
\begin{itemize}

\item The definition of $g$ and $g_1$ insures that in $r$'s support $con(Z'^k) \geq con(Z^k)$ for every $k \in [1,t]$: In case $Z$ and $Z'$ have the same consumption at round $k-1$ and $X(k)\neq N$, there is zero probability that $Y(k)=N$. The definition of $h$ and $h_1$ insures that in $r$'s support nowhere do $l \in [1,con(Z)]$ $[dig(Z)](l) = S$ and $[dig(Z')](l) = F$ hold simultaneously: If $[dig(X)](c+1)=S$, then with probability zero does $[dig(Y)](c+1)=Y(k)$ equal $F$. This is exactly our notion of history superiority, denoted by the relation $\succeq$. Consequently Criterion \ref{support} holds in the support of $f$.
\item As $q' > q$, $0 \leq h(Z,Z',k) \leq 1$, and as Criterion \ref{support} specifies $Z'^k \succeq Z^k$ whenever $con(Z'^k) = con(Z^k)$,  $0 \leq g(Z,Z',k) \leq 1$. Therefore $0 \leq R(Z,Z') \leq 1$ for every $Z, Z'$ that satisfy Criterion \ref{support}.
\end{itemize}

In summary, \eqref{f_def}  defines a matrix $f(Z,Z')$ of non-negative values, which may be non-zero only where Criterion \ref{support} holds, each row summing to the row marginal probability of $Z$.

We now define a second construction based on columns (values of $Z'$). The overall structure of this construction is similar to the previous one.

Let $Z, Z', X, Y \in \mathcal{H}_t$, and $0 < k \leq t$. We define below probability functions $g_2(X,Y,k), h_2(X,Y,k)$. Based on these, the probabilities $f'(Z,Z')$ will be given by:
\begin{align}
\label{g'_def}
& g'(X,Y,k) :=  \left[
       \begin{array}{ll}
	 g_2(X,Y,k)	& X(k) \neq N \\
	 1 - g_2(X,Y,k)	& X(k) = N
     \end{array}
    \right] \\
\label{h'_def}
& h'(X,Y,k) =  \left[
       \begin{array}{ll}
	 h_2(X,Y,k)	& X(k) = S \\
	 1 - h_2(X,Y,k)	& X(k) = F \\
	1	& X(k) = N
     \end{array}
    \right] \\
\label{r'_def}
& r'(X,Y,k) := g'(X,Y,k) h'(X,Y,k) \\
\label{R'_def}
& R'(Z,Z') := \prod\limits_{k=1}^t r'(Z,Z',k) \\
\label{f'_def}
& f'(Z,Z') := R'(Z,Z') \p[Z';q']
\end{align}

As before, Equations \eqref{g'_def} to \eqref{h2_def} yield $\sum r'(Z,Z',k)=1$ where the sum is over  $Z(k)=N, S$ and $F$. By \eqref{R'_def} it follows that:
\begin{equation}
\sum\limits_{Z \in \mathcal{H}_t} R'(Z,Z') = 1
\end{equation}

\noindent from which \eqref{column_marginal} follows, irrespective of the choice of $g_2$ or $h_2$, which we define next:

Mark $c := con(X^{k-1})$, $c' := con(Y^{k-1})$:

\begin{equation}
\label{g2_def}
g_2(X,Y,k) =  \left[
       \begin{array}{ll}
\left[
       \begin{array}{ll}
	0 & Y(k) = N \\
	\frac{\sigma(X^{k-1})}{\sigma(Y^{k-1})} & Y(k) \neq N \\
     \end{array}
    \right] & c = c' \\
	\sigma(X^{k-1}) & c \neq c' \\
     \end{array}
    \right]
\end{equation}

\begin{equation}
\label{h2_def}
h_2(X,Y,k) =  \left[
       \begin{array}{ll}
\left[
       \begin{array}{ll}
	0 & [dig(Y)](c+1) = F \\
	\frac{q}{q'}  & [dig(Y)](c+1) = S \\
     \end{array}
    \right] & c < con(X) \\
	1 & c \geq con(X) \\
     \end{array}
    \right]
\end{equation}

The choice of $g_2$ guarantees that in $r'$'s support the consumption requirement of Criterion \ref{support} holds, while the choice of $h_2$ guarantees that the monotonicity requirement of Criterion \ref{support} also holds. As $q' > q$, $0 \leq h'(Z,Z',k) \leq 1$, and from monotonicity it follows that $0 \leq g'(Z,Z',k) \leq 1$. Therefore $0 \leq R'(Z,Z') \leq 1$ for every $Z, Z'$ that satisfy Criterion \ref{support}.

In summary, \eqref{f'_def}  defines a matrix $f'(Z,Z')$ of non-negative values, which may be non-zero only where Criterion \ref{support} holds, each row summing to the column marginal probability of $Z$.

We now show that both constructions generate the same probabilities, $f(Z,Z') = f'(Z,Z')$, and that therefore the construction constitutes a joint probability distribution whose support satisfies Criterion \ref{support}:

We use \eqref{g_def} and \eqref{g'_def} to calculate $\frac{g(Z,Z',k)}{g'(Z,Z',k)}$. Where Criterion \ref{support} holds:\footnote{In the case $c=c', Z(k) \neq N, Z'(k) = N$ this expression evaluates to $\frac{0}{0}$, but this case is ruled out by Criterion \ref{support}.}
\begin{equation}
\label{g_ratio}
\frac{g(Z,Z',k)}{g'(Z,Z',k)} =  \left[
       \begin{array}{ll}
	\frac{1 - \sigma(Z'^{k-1})}{1 - \sigma(Z^{k-1})}  & Z'(k) = N, Z(k) = N \\
	\frac{1 - \sigma(Z'^{k-1})}{\sigma(Z^{k-1})}  & Z'(k) = N, Z(k) \neq N \\
	\frac{\sigma(Z'^{k-1})}{1- \sigma(Z^{k-1})} & Z'(k) \neq N, Z(k) = N \\
	\frac{\sigma(Z'^{k-1})}{\sigma(Z^{k-1})} & Z'(k) \neq  N, Z(k) \neq N
     \end{array}
    \right] = \frac{c(Z'^k)}{c(Z'^{k-1})}\Bigr/\frac{c(Z^k)}{c(Z^{k-1})}
\end{equation}

We define a pairing $[1,t] \mapsto [1,t]$, based on $Z, Z'$ as follows:
\begin{itemize}
\item The rounds where $Z$ has a consumption event ($Z(k) \neq N$) are paired with the corresponding rounds of the first $con(Z)$ consumption events in $Z'$. I.e. if $con(Z^{k-1}) = con(Z'^{k'-1}) = l-1, con(Z^k) = con(Z'^{k'}) = l$, then $k$ is paired with $k'$.
\item Any $con(Z') - con(Z)$ $N$-rounds of $Z$ are paired with rounds of the $[con(Z)+1]$'th to $con(Z')$'th consumption event in $Z'$.
\item The remaining $t - con(Z')$ rounds of $Z$ and $Z'$ are all $N$-rounds, and are paired in any order.
\end{itemize}

We use \eqref{h_def} and \eqref{h'_def} to calculate $\frac{h(Z,Z',k)}{h'(Z,Z',k')}$ where $k \rightarrow k'$ is part of the above pairing. Where Criterion \ref{support} holds:
\begin{equation}
\label{h_ratio}
\frac{h(Z,Z',k)}{h'(Z,Z',k')} =  \left[
       \begin{array}{ll}
\left[
       \begin{array}{ll}
	1 & Z(k) = N
     \end{array}
    \right] & Z'(k') = N \\
\left[
       \begin{array}{ll}
	\frac{1-q'}{1-q} & Z(k) = F \\
	1-q' & Z(k) = N
     \end{array}
    \right] & Z'(k') = F \\
\left[
       \begin{array}{ll}
	\frac{q'}{1-q} & Z(k) = F \\
	\frac{q'}{q} & Z(k) = S \\
	q' & Z(k) = N
     \end{array}
    \right] & Z'(k') = S
     \end{array}
    \right] = \frac{Q(Z'^{k'}; q')}{Q(Z'^{k'-1}; q')}\Bigr/\frac{Q(Z^k; q)}{Q(Z^{k-1}; q)}
\end{equation}

Combining \eqref{R_def}, \eqref{R'_def}, \eqref{g_ratio} and \eqref{h_ratio}, we calculate:
\begin{equation}
\frac{R(Z,Z')}{R'(Z,Z')} = \frac{c(Z')}{c(Z)} \frac{Q(Z'; q')}{Q(Z; q)} = \frac{\p[Z'; q']}{\p[Z; q]}
\end{equation}

Recalling \eqref{f_def}, \eqref{f'_def} we conclude $f(Z,Z') = f'(Z,Z')$. Therefore the two constructions are the same, as we sought to show, and therefore form a joint distribution with support given by Criterion \ref{support}.

The theorem now follows: In the joint distribution matrix, consider a minor with rows for which $con(Z) \geq x$, and columns for which $con(Z') \geq x$. Since Criterion \ref{support} requires $con(Z') \geq con(Z)$, the minor contains all the support of each included row. Therefore:
\begin{equation}
\label{ineq}
\sum\limits_{con(Z) \geq x} \p[Z; q] \leq \sum\limits_{con(Z') \geq x} \p[Z'; q']
\end{equation}

Since this is true for any $q' > q$, $\p \Bigl[con(Z) \geq x | Z \in \mathcal{H}_t \Bigr]$ is non-decreasing in the product quality, as the theorem asserts.
\qed
\end{proof}

It follows from Theorem \ref{theorem1} that the expected consumption is non-decreasing in the quality:

\begin{corollary}
\label{expectcor}
If the partiality strategy is monotone, then for any time $t$:
\begin{equation}
\frac{d}{dq} \E \Bigl[con(Z) | Z \in \mathcal{H}_t \Bigr] \geq 0
\end{equation}
where the probability space is  $\mathcal{H}$.
\end{corollary}

\begin{proof}
Since:
\begin{equation}
\E \Bigl[con(Z) | Z \in \mathcal{H}_t \Bigr] = \sum\limits_{x=1}^\infty\p\Bigl[con(Z) \geq x | Z \in \mathcal{H}_t \Bigr]
\end{equation}
This follows from Theorem \ref{theorem1}.
\end{proof}

Having proved Theorem \ref{theorem1} for the consumption of a single customer, we state and prove an equivalent theorem for an entire market, i.e. that, provided each customer's partiality strategy is monotone, a product's market share (here defined as its total consumption) stochastically dominates the market share of a product of lesser quality:

\begin{theorem}
\label{market_history}
Let each of $n$ customers have a monotone partiality strategy. Given a time $t$ and an integer $x$:
\begin{equation}
\frac{d}{dq} \p \Bigl[\sum\limits_{j=1}^n con(Z_j) \geq x | \forall{j \in [n]}, Z_j \in \mathcal{H}_t \Bigr] \geq 0
\end{equation}
where the probability space is  $\mathcal{H}^n$.
\end{theorem}

\begin{proof}
The proof proceeds by induction on $n$. For $n=1$ this is just Theorem \ref{theorem1}. Assume the statement true for up to $n-1$ customers. Mark $a_n(x) := \p \Bigl[\sum\limits_{j=1}^n con(Z_j) \geq x | \forall{j \in [n]}, Z_j \in \mathcal{H}_t \Bigr]$ and $b(y) := \p \Bigl[Z_n \geq y | Z_n \in \mathcal{H}_t \Bigr]$. Then:
\begin{align*}
\frac{d}{dq} a_n(x) &=  \frac{d}{dq} \sum\limits_{y=-\infty}^\infty a_{n-1}(x-y) \p \Bigl[Z_n = y | Z_n \in \mathcal{H}_t \Bigr] = \\
& =\frac{d}{dq} \sum\limits_{y=-\infty}^\infty a_{n-1}(x-y) b(y) - \frac{d}{dq} \sum\limits_{y=-\infty}^\infty a_{n-1}(x-y) b(y+1) = \\
& = \sum\limits_{y=-\infty}^\infty  \frac{d}{dq} a_{n-1}(x-y) b(y) + \sum\limits_{y=-\infty}^\infty a_{n-1}(x-y)  \frac{d}{dq} b(y) - \\
&  - \sum\limits_{y=-\infty}^\infty  \frac{d}{dq} a_{n-1}(x-y)  b(y+1)  - \sum\limits_{y=-\infty}^\infty a_{n-1}(x-y)   \frac{d}{dq}b(y+1)
\end{align*}

Changing variables in the last term $y+1 \rightarrow y$ and combining, this results in:
\begin{align*}
& \sum\limits_{y=-\infty}^\infty  \frac{d}{dq}  a_{n-1}(x-y) \p \Bigl[Z_n = y | Z_n \in \mathcal{H}_t \Bigr] + \\
& \sum\limits_{y=-\infty}^\infty \p \Bigl[\sum\limits_{j=1}^{n-1} con(Z_j) = x-y | \forall{j \in [n-1]}, Z_j \in \mathcal{H}_t \Bigr]  \frac{d}{dq} b(y) \geq 0
\end{align*}
since all factors in the above expression are non-negative.
\end{proof}

The main result can now be stated and proved: If two products are interchangeable in the eyes of the customers, and if there is no prior cause to believe that one of the products has the better quality, then from the observation of a higher market share for one of the products one can infer that it has the better quality.
\begin{theorem}
\label{theorem-history}
Let the partiality strategies of all customers for products $1$ and $2$ be monotone, and let each customer's strategy for product $1$ be the same as for product $2$. Let products $1,2$ have possibly different qualities $q_1, q_2$ respectively, with symmetric prior. Let the history of customer $j \in [n]$ with product $i \in [m]$ be $Z_{ij} \in \mathcal{H}_t$. Then:
\begin{equation*}
\p\Bigl[q_1 \geq q_2 | \sum\limits_{j=1}^n con(Z_{1j}) > \sum\limits_{j=1}^n con(Z_{2j})\Bigr] \geq \p\Bigl[q_2 \geq q_1 | \sum\limits_{j=1}^n con(Z_{1j}) > \sum\limits_{j=1}^n con(Z_{2j})\Bigr]
\end{equation*}
where the probability space is  $\mathcal{H}^{mn}$.
\end{theorem}

\begin{proof}
Mark $\omega_1 := \sum\limits_{j=1}^n con(Z_{1j})$ and $\omega_2 := \sum\limits_{j=1}^n con(Z_{2j})$.

As the products are interchangeable $\p[\omega_1 > \omega_2 | q_1=q_2] = \p[\omega_2 > \omega_1 | q_1=q_2]$. By theorem \ref{market_history}:
\begin{align}
\p[\omega_1 > \omega_2 | q_1 \geq q_2] \geq \p[\omega_1 > \omega_2 | q_1 = q_2] \\
\p[\omega_2 > \omega_1 | q_1 \geq q_2] \leq \p[\omega_2 > \omega_1 | q_1 = q_2]
\end{align}

Therefore:
\begin{equation}
\p[\omega_1 > \omega_2 | q_1 \geq q_2] \geq \p[\omega_2 > \omega_1 | q_1 \geq q_2]
\end{equation}

As the products are interchangeable, $\p[\omega_2 > \omega_1 | q_1 \geq q_2] = \p[\omega_1 > \omega_2 | q_2 \geq q_1]$, therefore:
\begin{equation}
\p[\omega_1 > \omega_2 | q_1 \geq q_2] \geq \p[\omega_1 > \omega_2 | q_2 \geq q_1]
\end{equation}

From which the theorem follows by Bayes' theorem and the symmetric prior on $q_1, q_2$.
\qed
\end{proof}

\section{Model with Market Share Observations}
\label{themodelplus}

We now generalize our model to the case where the market share of the products is known to customers. Customers can base their partiality strategies on market share information, as well as on their individual history with the products.
We need to define market share:

Let there be $n$ customers and $m$ products. $q_i$ denotes product $i$'s quality. Let $Z_{ij} \in \mathcal{H}_t$ be customer $j$'s $t$-deep history with product $i$. We define a ($t$-deep) {\em history ensemble} $Z$ as a set of histories for each customer-product combination $Z := \{Z_{ij} \in \mathcal{H}_t, \forall{i \in [m], j \in [n]}\}$.

The set of $t$-deep history ensembles is denoted by $\mathcal{G}_t$.

Given a $t$-deep history ensemble $Z$, and an initial market share $A := (A_1, \ldots, A_m)$, the market share of product $i \in [m]$ after round $k$ is the total number of units consumed of product $m$ up to round $k$, and is denoted by $\Omega_i(Z,A,k)$:
\begin{equation}
\label{market-share-def}
\Omega_i(Z,A,k) = A_i + \sum\limits_{j=1}^n con(Z_{ij}^k)
\end{equation}

All customers are aware of the round-$k$ market share of all products when they make their consumption decisions at round $k+1$.\footnote{See the Discussion for some comments on how things behave when customers may have more information about past market shares.}

$\Omega(Z,A,k)$ denotes the vector of all product market shares \\ $(\Omega_1(Z,A,k), \ldots, \Omega_m(Z,A,k))$.

The initial market share vector $A \equiv \Omega(Z,A,0)$ is the market share vector before round $1$, and its value is extraneous to the model.

The partiality strategy of customer $j \in [n]$ to product $i \in [m]$ after round $t$, \\ $\sigma_{ij}(Z_{ij},\Omega(Z,A,t))$, is defined as the probability that the customer will consume product $i$ at round $t+1$. It is a behavioral strategy that depends on the customer's information set which consists of $Z_{ij}$, the customer's history with product $i$, and the market share $\Omega(Z,A,t)$ known after round $t$.

Mark $q := (q_1, \ldots, q_m)$. Define:
\begin{align}
S_i(Z) & := \sum_{j=1}^n S(Z_{ij}) \\
F_i(Z) & := \sum_{j=1}^n F(Z_{ij}) \\
Q_i(Z;q) & := q_i^{S_i(Z)} (1-q_i)^{F_i(Z)} \\
Q(Z;q) & := \prod\limits_{i=1}^m Q_i(Z;q)
\end{align}

We define monotonicity similarly to how we defined it in Section \ref{themodel}: A partiality strategy $\sigma_{ij}$ is monotone if for every market share vector $\omega$ and for every history pair $Z_1, Z_2$ satisfying $Z_1 \succeq Z_2$, $\sigma_{ij}(Z_1,\omega) \geq \sigma_{ij}(Z_2,\omega)$.

We introduce a condition on customers' response to market data that we call {\em weak herding}. As we show, if weak herding holds and if customers are monotone, then a result similar to Theorem \ref{theorem-history} holds. Namely, market share still indicates quality.

A customer $j$ is called {\em weakly herding} if for every product $i$, time $t$, $t$-deep history $Z_{ij}$ and market share vector $\omega = \{\omega_1, \ldots, \omega_m\}$, $\sigma_{ij}(Z_{ij},\omega)$ is non-decreasing in $\omega_i$ (and is independent of $\omega_k$ for $k \neq i$). A customer $j$ is called {\em competitively weakly herding} if for every product $i$, time $t$, $t$-deep history $Z_{ij}$ and market share vector $\omega = \{\omega_1, \ldots, \omega_m\}$, $\sigma_{ij}(Z_{ij},\omega)$ is non-decreasing in $\omega_i$ {\em and} non-increasing in $\omega_k$ for all $k \neq i$. In particular, a customer who, as in our basic model, is unaware of market share or disregards it, is both weakly herding and competitively weakly herding.

Weak herding is a natural response to market share data: The more a product has been consumed, the more a customer who is aware of that fact is disposed to consume it. Competitive weak herding makes it possible to base partiality strategies on a product's {\em order} in market share data, e.g. on whether or not a product is the market leader in consumption. We shall be able to prove our thesis for both kinds of responses, though for competitively weak herding we shall have to limit the number of products.

Additionally, we define {\em anonymity} for products, the property that products are {\em a priori} equal in the eyes of customers. For anonymous products, partiality strategies do not depend on the label of a product but only on its data. Formally, let $\omega = (\omega_1, \omega_2, \ldots, \omega_m)$. Define the permutation $K_{12}(\omega) := (\omega_2, \omega_1, \omega_3, \ldots, \omega_m)$. Then products $1, 2$ are {\em anonymous} if, for each customer $j$, for each history $Z \in \mathcal{H}$ and for each market share vector $\omega$:
\begin{equation}
\sigma_{1j}(Z,\omega) = \sigma_{2j}(Z,K_{12}(\omega))
\end{equation}

\section{Theorem with Market Share Observation}
\label{marketshare}

\begin{theorem}
\label{theorem2}
\begin{enumerate}
\item Assume that all customers are monotone and weakly herding and fix some initial conditions $A = (A_1, \ldots, A_m)$. Then for all times $t$:
\begin{equation}
\frac{d}{dq_1} \p \Bigl[\Omega_1(Z,A,t) > \Omega_2(Z,A,t) | Z \in \mathcal{G}_t \Bigr] \geq 0
\end{equation}
where $q_i \in [0,1]$ is the quality of product $i$ for $i=1,\ldots,m$, and the probability space is $\mathcal{G}$.
\item The same holds when all customers are competitively weakly herding, rather than weakly herding, and there are two products ($m=2$).
\end{enumerate}
\end{theorem}

\begin{proof}
Our proof is modeled on the proof of Theorem \ref{theorem1}, with changes:

Let $q := (q_1, q_2, \ldots, q_m)$ and $q' := (q'_1, q'_2, \ldots, q'_m)$ be two product quality vectors satisfying $q'_1 > q_1$ and $q_2 = q'_2, \ldots, q_m = q'_m$, so that $q'$ is different from $q$ only in having better quality for product $1$. For the proof, we will construct a  joint distribution of two types of history ensembles:
\begin{enumerate}
\item The distribution of history ensembles $\in \mathcal{G}_t$ under product qualities $q$. Ensembles from this distribution are denoted by Z.
\item The distribution of history ensembles $\in \mathcal{G}_t$ under product qualities $q'$. Ensembles from this distribution are denoted by Z'.
\end{enumerate}

As before, the joint distribution is a proof technique and not a representation of any real customer behavior.

The construction assigns to each pair of ensembles $Z, Z' \in \mathcal{G}_t$ a joint probability $f(Z,Z') \in [0,1]$. Assuming $f(Z,Z')$ to be tabulated in a table whose rows correspond to values of $Z$ and columns correspond to values of $Z'$, a valid joint distribution must have rows summing to the correct ensemble event probabilities in both rows and columns, i.e. for every $Z \in \mathcal{G}_t$:
\begin{equation}
\label{row_e_marginal}
\sum\limits_{Z' \in \mathcal{G}_t} f(Z,Z') = \p[Z;q] = c(Z) Q(Z; q)
\end{equation}

\noindent and for every $Z' \in \mathcal{G}_t$:
\begin{equation}
\label{column_e_marginal}
\sum\limits_{Z \in \mathcal{G}_t} f(Z,Z') = \p[Z';q'] = c(Z') Q(Z'; q')
\end{equation}

Our construction restricts the joint distribution's support (i.e. those pairs $Z, Z'$ where $f(Z,Z') > 0$) by the following criterion:

\begin{criterion}
\label{support_e}
For every customer $j \in [n]$, and round $0 < k \leq t$:
\begin{itemize}
\item There exists $0 < l \leq k$ s.t. $dig({Z'}_{1j}^l) \succeq dig(Z_{1j}^k)$.
\item For every product $i \in [m]$, $i \neq 1$,  there exists $0 < l \leq k$ s.t. $dig(Z_{ij}^l) = dig({Z'}_{ij}^k)$. (In the case of weak herding, {\em a fortiori}, $Z_{ij}^k = {Z'}_{ij}^k$).
\end{itemize}

The following statements are all consequences of the criterion:
\begin{itemize}
\item $con({Z'}_{1j}^k) \geq con(Z_{1j}^k)$, while for every $i \neq 1$, $con({Z'}_{ij}^k) \leq con(Z_{ij}^k)$
\item $con({Z'}_{1j}^k) = con(Z_{1j}^k) \Rightarrow {Z'}_{1j}^k \succeq Z_{1j}^k$,  while for every $i \neq 1$, $con({Z'}_{ij}^k) = con(Z_{ij}^k) \Rightarrow Z_{ij}^k \succeq {Z'}_{ij}^k$.
\item $\Omega_1(Z',A,k) \geq \Omega_1(Z,A,k)$, while for every $i \neq 1$, $\Omega_i(Z',A,k) \leq \Omega_i(Z,A,k)$.
\item If monotonicity and either weak herding or competitively weak herding hold, \\ $\sigma_{1j}(Z',\Omega(Z',A,k)) \geq \sigma_{1j}(Z,\Omega(Z,A,k))$.
\item If monotonicity and either weak herding (for any number of products) or competitively weak herding (for two products) hold, $\sigma_{ij}(Z',\Omega(Z',A,k)) \leq \sigma_{ij}(Z,\Omega(Z,A,k))$. Note that it is here that the case $m > 2$ of competitively weak herding breaks, since the presence of a 3rd product, whose market share possibly decreased from $Z$ to $Z'$, would not allow drawing this consequence.
\end{itemize}
\end{criterion}

We specify the joint distribution in two different ways: In the first, it is specified how the row marginal probabilities, i.e. probabilities of $Z$ events, are split into the row's individual probabilties. The construction will explicitly guarantee \eqref{row_e_marginal} and Criterion \ref{support_e}.  In the second, it is specified how the column marginal probabilities, i.e. probabilities of $Z'$ events, are split into the column's individual probabilties. The construction will explicitly guarantee \eqref{column_e_marginal} and Criterion \ref{support_e}. Finally, we will demonstrate that the two constructions coincide and lead to the same probabilities, thus proving that it is in fact a joint probability fulfilling \eqref{row_e_marginal}, \eqref{column_e_marginal} and Criterion \ref{support_e}.

The first, row construction, is specified by:

Let $Z, Z',X,Y \in \mathcal{G}_t$ (to improve readability $X$ will be used in place of $Z$, and $Y$ in place of $Z'$), $i \in [m]$, $j \in [n]$ and $0 < k \leq t$. We will define below probability functions $g_3(X,Y,i,j,k), h_3(X,Y,i,j,k)$. Based on these, the probabilities $f(Z,Z')$ will be given by:
\begin{align}
\label{g_e_def}
& g(X,Y,i,j,k) :=  \left[
       \begin{array}{ll}
	 g_3(X,Y,i,j,k)	& Y_{ij}(k) \neq N \\
	 1 - g_3(X,Y,i,j,k)	& Y_{ij}(k) = N
     \end{array}
    \right] \\
\label{h_e_def}
& h(X,Y,i,j,k) =  \left[
       \begin{array}{ll}
	 h_3(X,Y,i,j,k)	& Y_{ij}(k) = S \\
	 1 - h_3(X,Y,i,j,k)	& Y_{ij}(k) = F \\
	1	& Y_{ij}(k) = N
     \end{array}
    \right] \\
\label{r_e_def}
& r(X,Y,i,j,k) := g(X,Y,i,j,k) h(X,Y,i,j,k) \\
\label{R_e_def}
& R(Z,Z') := \prod\limits_{i=1}^m \prod\limits_{j=1}^n \prod\limits_{k=1}^t r(Z,Z',i,j,k) \\
\label{f_e_def}
& f(Z,Z') := R(Z,Z') \p[Z;q]
\end{align}

From Equations \eqref{g_e_def} to \eqref{r_e_def} we derive:
\begin{align}
r(X,Y,i,j,k) =  \left[
       \begin{array}{ll}
	 g(X,Y,i,j,k) h(X,Y,i,j,k)	& Y_{ij}(k) = S \\
	 g(X,Y,i,j,k) [1 - h(X,Y,i,j,k)]	& Y_{ij}(k) = F \\
	1-g(X,Y,i,j,k)	& Y_{ij}(k) = N
     \end{array}
    \right]
\end{align}

\noindent The three values of which sum to $1$. It follows that for every given values of $Z_{ij}$ and ${Z'}_{ij}^{k-1}$ there holds  $\sum r(Z,Z',i,j,k) = 1$ where the sum is over $Z'_{ij} = N, S$ and $F$. By \eqref{R_e_def} it follows that:
\begin{equation}
\sum\limits_{Z' \in \mathcal{G}_t} R(Z,Z') = 1
\end{equation}

\noindent from which \eqref{row_e_marginal} follows. Note that this holds regardless of the specific choice of $g_3, h_3$, which come in next.

Mark $c := con(X_{ij}^{k-1})$, $c' := con(Y_{ij}^{k-1})$, $\Sigma(Z,k) := \sigma_{ij}(Z_{ij}^k,\Omega(Z,A,k))$

\begin{equation}
\label{g3_def}
g_3(X,Y,i,j,k) =  \left[
       \begin{array}{ll}
\left[
       \begin{array}{ll}
\left[
       \begin{array}{ll}
	1 & X_{ij}(k) \neq N \\
	\frac{\Sigma(Y,k-1) - \Sigma(X,k-1)}{1 - \Sigma(X,k-1)} & X_{ij}(k) = N \\
     \end{array}
    \right] & i = 1 \\
\left[
       \begin{array}{ll}
	0 & X_{ij}(k) = N \\
	\frac{\Sigma(Y,k-1)}{\Sigma(X,k-1)} & X_{ij}(k) \neq N \\
     \end{array}
    \right] & i \neq 1
     \end{array}
    \right] & c = c' \\
	\Sigma(Y,k-1) & c \neq c' \\
     \end{array}
    \right]
\end{equation}

\begin{equation}
\label{h31_def}
h_{31}(X,Y,i,j,k) =  \left[
       \begin{array}{ll}
\left[
       \begin{array}{ll}
	1 & [dig(X_{ij})](c+1) = S \\
	\frac{q'_i-q_i}{1-q_i}  & [dig(X_{ij})](c+1) = F \\
     \end{array}
    \right] & c < con(X_{ij}) \\
	q'_i  & c \geq con(X_{ij}) \\
     \end{array}
    \right]
\end{equation}

\begin{equation}
\label{h3i_def}
h_{3i}(X,Y,i,j,k) =  \left[
       \begin{array}{ll}
\left[
       \begin{array}{ll}
	1 & [dig(X_{ij})](c+1) = S \\
	\frac{q'_i-q_i}{1-q_i}  & [dig(X_{ij})](c+1) = F \\
     \end{array}
    \right] & c < con(Y_{ij}) \\
	1  & c \geq con(Y_{ij}) \\
     \end{array}
    \right]
\end{equation}

\begin{equation}
\label{h3_def}
h_3(X,Y,i,j,k) =  \left[
       \begin{array}{ll}
	h_{31}(X,Y,i,j,k)	& i  = 1 \\
	h_{3i}(X,Y,i,j,k)	& i  \neq 1
     \end{array}
    \right]
\end{equation}

An examination of \eqref{g_e_def} to \eqref{h3_def} leads to the following observations:
\begin{itemize}
\item The definition of $g$ and $g_3$ insures that in $r$'s support $con({Z'}_{1j}^k) \geq con(Z_{1j}^k)$ for every $j \in [n]$, $k \in [1,t]$, and $con({Z'}_{ij}^k) \leq con(Z_{ij}^k)$ for every $i \in [2,m]$, $j \in [n]$, $k \in [1,t]$. The definition of $h$ and $h_3$ insures that in $r$'s support nowhere do $l \in [1,con(Z)]$ $[dig(Z)](l) = S$ and $[dig(Z')](l) = F$ hold simultaneously. This is the definition of history superiority, denoted by the relation $\succeq$. Consequently Criterion \ref{support_e} holds in the support of $f$.
\item The difference between \eqref{h31_def} and \eqref{h3i_def} (compare with \eqref{h1_def} and \eqref{h2_def}) stems from the fact that for product $1$, $X$'s digest is shorter or equal to $Y$'s, while for other products, it is $Y$'s digest that is shorter or equal.
\item As $q_i' \geq q_i$ for every $i \in [m]$, $0 \leq h(Z,Z',i,j,k) \leq 1$. Also, as a consequence of Criterion \ref{support_e} is that $\Sigma(Z',k-1) \geq \Sigma(Z,k-1)$ for product $1$ ($i = 1$), and $\Sigma(Z',k-1) \leq \Sigma(Z,k-1)$ for other products ($i \neq 1$),  $0 \leq g(Z,Z',i,j,k) \leq 1$. Therefore $0 \leq R(Z,Z') \leq 1$ for every $Z, Z'$ satisfying Criterion \ref{support_e}.
\end{itemize}

In summary, \eqref{f_e_def}  defines a matrix $f(Z,Z')$ of non-negative values, which may be non-zero only where Criterion \ref{support_e} holds, each row summing to the row marginal probability of $Z$.

We now define a second construction based on columns (values of $Z'$):

Let $Z, Z',X,Y \in \mathcal{G}_t$, $i \in [m]$, $j \in [n]$ and $0 < k \leq t$. We define below probability functions $g_4(X,Y,i,j,k), h_4(X,Y,i,j,k)$. Based on these, the probabilities $f(Z,Z')$ will be given by:
\begin{align}
\label{g'_e_def}
& g'(X,Y,i,j,k) :=  \left[
       \begin{array}{ll}
	 g_4(X,Y,i,j,k)	& X_{ij}(k) \neq N \\
	 1 - g_4(X,Y,i,j,k)	& X_{ij}(k) = N
     \end{array}
    \right] \\
\label{h'_e_def}
& h'(X,Y,i,j,k) =  \left[
       \begin{array}{ll}
	 h_4(X,Y,i,j,k)	& X_{ij}(k) = S \\
	 1 - h_4(X,Y,i,j,k)	& X_{ij}(k) = F \\
	1	& X_{ij}(k) = N
     \end{array}
    \right] \\
\label{r'_e_def}
& r'(X,Y,i,j,k) := g'(X,Y,i,j,k) h'(X,Y,i,j,k) \\
\label{R'_e_def}
& R'(Z,Z') := \prod\limits_{i=1}^m \prod\limits_{j=1}^n \prod\limits_{k=1}^t r'(Z,Z',i,j,k) \\
\label{f'_e_def}
& f'(Z,Z') := R'(Z,Z') \p[Z';q']
\end{align}

An examination of \eqref{g'_e_def} to \eqref{h4_def} leads to observations similar to those noted above for the first construction: There holds $\sum r'(Z,Z',k)=1$ where the sum is over $Z_{ij}(k) = N,S$ and $F$. By \eqref{R'_e_def} it follows that:
\begin{equation}
\sum\limits_{Z \in \mathcal{G}_t} R'(Z,Z') = 1
\end{equation}

\noindent from which \eqref{column_e_marginal} follows.
We define the functions $g_4$ and $h_4$:

Mark $c := con(X_{ij}^{k-1})$, $c' := con(Y_{ij}^{k-1})$, $\Sigma(Z,k) := \sigma_{ij}(Z_{ij}^k,\Omega(Z,A,k))$

\begin{equation}
\label{g4_def}
g_4(X,Y,i,j,k) =  \left[
       \begin{array}{ll}
\left[
       \begin{array}{ll}
\left[
       \begin{array}{ll}
	0 & Y_{ij}(k) = N \\
	\frac{\Sigma(X,k-1)}{\Sigma(Y,k-1)} & Y_{ij}(k) \neq N \\
     \end{array}
    \right] & i = 1 \\
\left[
       \begin{array}{ll}
	1 & Y_{ij}(k) \neq N \\
	\frac{\Sigma(X,k-1) - \Sigma(Y,k-1)}{1 - \Sigma(Y,k-1)} & Y_{ij}(k) = N \\
     \end{array}
    \right] & i \neq 1
     \end{array}
    \right] & c = c' \\
	\Sigma(X,k-1) & c \neq c' \\
     \end{array}
    \right]
\end{equation}

\begin{equation}
\label{h41_def}
h_{41}(X,Y,i,j,k) =  \left[
       \begin{array}{ll}
\left[
       \begin{array}{ll}
	0 & [dig(Y_{ij})](c+1) = F \\
	\frac{q_i}{q'_i}  & [dig(Y_{ij})](c+1) = S \\
     \end{array}
    \right] & c < con(X_{ij}) \\
	1  & c \geq con(X_{ij}) \\
     \end{array}
    \right]
\end{equation}

\begin{equation}
\label{h4i_def}
h_{4i}(X,Y,i,j,k) =  \left[
       \begin{array}{ll}
\left[
       \begin{array}{ll}
	0 & [dig(Y_{ij})](c+1) = F \\
	\frac{q_i}{q'_i}  & [dig(Y_{ij})](c+1) = S \\
     \end{array}
    \right] & c < con(Y_{ij}) \\
	q_i  & c \geq con(Y_{ij}) \\
     \end{array}
    \right]
\end{equation}

\begin{equation}
\label{h4_def}
h_4(X,Y,i,j,k) =  \left[
       \begin{array}{ll}
	h_{41}(X,Y,i,j,k)	& i  = 1 \\
	h_{4i}(X,Y,i,j,k)	& i  \neq 1
     \end{array}
    \right]
\end{equation}

The choice of $g_4$ guarantees that in $r'$'s support the market share requirements of Criterion \ref{support_e} holds, while the choice of $h_4$ guarantees that the monotonicity requirements of Criterion \ref{support_e} also holds. As $q'_i \geq q_i$, $0 \leq h'(Z,Z',i,j,k) \leq 1$, and from monotonicity and (competitive) weak herding it follows that $0 \leq g'(Z,Z',i,j,k) \leq 1$. Therefore $0 \leq R'(Z,Z') \leq 1$ for every $Z, Z'$ satisfying Criterion \ref{support_e}.

In summary, \eqref{g'_e_def} through \eqref{h4_def} define a matrix $f'(Z,Z')$ of non-negative values, which may be non-zero only where Criterion \ref{support_e} holds, each column summing to the column marginal probability of $Z'$.

We now show that both constructions generate the same probabilities, $f(Z,Z') = f'(Z,Z')$, and that therefore the construction is for a joint probability with Criterion \ref{support_e} for its support:

We use \eqref{g_e_def} and \eqref{g'_e_def} to calculate that where $Z, Z'$ obey Criterion \ref{support_e}:\footnote{i.e. except for the cases $c=c', i=1, Z(k) \neq N, Z'(k) = N$ and $c=c', i \neq 1, Z(k) = N, Z'(k) \neq N$, where the expression evalutes to $\frac{0}{0}$}
\begin{equation}
\label{g_e_ratio}
\frac{g(Z,Z',i,j,k)}{g'(Z,Z',i,j,k)} =  \left[
       \begin{array}{ll}
	\frac{1 - \Sigma(Z',k-1)}{1 - \Sigma(Z,k-1)}  & Z'_{ij}(k) = N, Z_{ij}(k) = N \\
	\frac{1 - \Sigma(Z',k-1)}{\Sigma(Z,k-1)}  & Z_{ij}'(k) = N, Z_{ij}(k) \neq N \\
	\frac{\Sigma(Z',k-1)}{1- \Sigma(Z,k-1)} & Z'_{ij}(k) \neq N, Z_{ij}(k) = N \\
	\frac{\Sigma(Z',k-1)}{\Sigma(Z,k-1)} & Z'_{ij}(k) \neq  N, Z_{ij}(k) \neq N
     \end{array}
    \right] = \frac{c_{ij}({Z'}_{ij}^k)}{c_{ij}({Z'}_{ij}^{k-1})}\Bigr/\frac{c_{ij}(Z_{ij}^k)}{c_{ij}(Z_{ij}^{k-1})}
\end{equation}

For each $i \in [m]$, $j \in [n]$, we define a pairing $[1,t] \mapsto [1,t]$, based on $Z_{ij}, Z'_{ij}$. The following describes the pairing for product $1$:
\begin{itemize}
\item The rounds where $Z_{ij}$  has a consumption event ($Z_{ij}(k) \neq N$) are paired with the corresponding rounds of the first $con(Z_{ij})$ consumption events in $Z'_{ij}$. I.e. if $con(Z_{ij}^{k-1}) = con({Z'}_{ij}^{k'-1}) = l-1, con(Z_{ij}^k) = con({Z'}_{ij}^{k'}) = l$, then $k$ is paired with $k'$.
\item Any $con(Z'_{ij}) - con(Z_{ij})$ $N$-rounds of $Z_{ij}$ are paired with rounds of the $[con(Z_{ij})+1]$'th to $con(Z'_{ij})$'th consumption event in $Z'_{ij}$.
\item The remaining $t - con(Z'_{ij})$ rounds of $Z_{ij}$ and $Z'_{ij}$ are all $N$-rounds, and are paired in any order.
\end{itemize}

For products other than $1$, use the above pairing, but exchange $Z_{ij}$ for $Z'_{ij}$, and vice versa (the reason being that it is $Z'_{ij}$ instead of $Z_{ij}$ that has the shorter digest).

We use \eqref{h_e_def} and \eqref{h'_e_def} to calculate that where $Z, Z'$ obey Criterion \ref{support_e}, and $k \rightarrow k'$ is part of the above pairing:
\begin{equation}
\label{h_e_ratio}
\frac{h(Z,Z',i,j,k')}{h'(Z,Z',i,j,k)} =  \left[
       \begin{array}{ll}
\left[
       \begin{array}{ll}
	\frac{1}{1-q_i} & Z_{ij}(k) = F \\
	\frac{1}{q_i} & Z_{ij}(k) = S \\
	1 & Z_{ij}(k) = N
     \end{array}
    \right] & Z_{ij}'(k') = N \\
\left[
       \begin{array}{ll}
	\frac{1-q'_i}{1-q_i} & Z_{ij}(k) = F \\
	\frac{1-q'_i}{q_i} & Z_{ij}(k) = S \\
	1-q'_i & Z_{ij}(k) = N
     \end{array}
    \right] & Z'_{ij}(k') = F \\
\left[
       \begin{array}{ll}
	\frac{q'_i}{1-q_i} & Z_{ij}(k) = F \\
	\frac{q'_i}{q_i} & Z_{ij}(k) = S \\
	q'_i & Z_{ij}(k) = N
     \end{array}
    \right] & Z'_{ij}(k') = S
     \end{array}
    \right] = \frac{Q({Z'}_{ij}^{k'}; q'_i)}{Q({Z'}_{ij}^{k'-1}; q'_i)}\Bigr/\frac{Q(Z_{ij}^k; q_i)}{Q(Z_{ij}^{k-1}; q_i)}
\end{equation}

Combining \eqref{R_e_def}, \eqref{R'_e_def}, \eqref{g_e_ratio} and \eqref{h_e_ratio}, we calculate:
\begin{equation}
\frac{R(Z,Z')}{R'(Z,Z')} = \frac{c(Z')}{c(Z)} \frac{Q(Z'; q')}{Q(Z; q)} = \frac{\p[Z'; q']}{\p[Z; q]}
\end{equation}

Recalling \eqref{f_e_def}, \eqref{f'_e_def} we conclude $f(Z,Z') = f'(Z,Z')$. Therefore the two constructions are the same, as we sought to show, and therefore form a joint distribution with support given by Criterion \ref{support_e}.

The theorem now follows: In the joint distribution matrix, consider a minor with rows for which $\Omega_1(Z,A,t) > \Omega_2(Z,A,t)$, and columns for which $\Omega_1(Z',A,t) > \Omega_2(Z',A,t)$. Since a consequence of Criterion \ref{support_e} is that $\Omega_1(Z',A,t) \geq \Omega_1(Z,A,t)$ and $\Omega_2(Z',A,t) \leq \Omega_2(Z,A,t)$, the minor contains all the support of each included row. Therefore:
\begin{equation}
\label{ineq_e}
\sum\limits_{\Omega_1(Z,A,t) > \Omega_2(Z,A,t) } \p[Z; q] \leq \sum\limits_{\Omega_1(Z',A,t) > \Omega_2(Z',A,t) } \p[Z'; q']
\end{equation}

Since this is true for any $q'_1 > q_1$,  $\p \Bigl[\Omega_1(Z,A,t) > \Omega_2(Z,A,t) | Z \in \mathcal{G}_t \Bigr]$  is non-decreasing in $q_1$, as the theorem asserts.
\qed
\end{proof}

We now state and prove the main result:

\begin{theorem}
\label{theorem-market-share}
\begin{enumerate}
\item Let there be $m$ products, and let products $1,2$ be anonymous but have possibly different qualities $q_1, q_2$ respectively, with symmetric prior on their quality and initial market share. Let all customers have monotone partiality strategies to these products, and let all customers be weakly herding. Let $\omega_1, \omega_2$ be the observed market share after time $t$ of $1$ and $2$, respectively. Then:
\begin{equation}
\p[q_1 \geq q_2 | \omega_1 > \omega_2] \geq \p[q_2 \geq q_1 | \omega_1 > \omega_2]
\end{equation}
\item The same holds when all customers are competitively weakly herding, rather than weakly herding, and there are two products ($m=2$).
\end{enumerate}\end{theorem}

\begin{proof}
The observed market share is the result of some history ensemble $Z \in \mathcal{G}_t$ and some initial market share vector $A = (A_1, \ldots, A_m)$, such that:
\begin{align}
\omega_1 = \Omega_1(Z,A,t) \\
\omega_2 = \Omega_2(Z,A,t)
\end{align}

As the two products are anonymous, and the market share prior is symmetric, we must have $\p[\omega_1 > \omega_2 | q_1=q_2] = \p[\omega_2 > \omega_1 | q_1=q_2]$. By theorem \ref{theorem2}:
\begin{align}
\p[\omega_1 > \omega_2 | q_1 \geq q_2] \geq \p[\omega_1 > \omega_2 | q_1 = q_2] \\
\p[\omega_2 > \omega_1 | q_1 \geq q_2] \leq \p[\omega_2 > \omega_1 | q_1 = q_2]
\end{align}

Therefore:
\begin{equation}
\p[\omega_1 > \omega_2 | q_1 \geq q_2] \geq \p[\omega_2 > \omega_1 | q_1 \geq q_2]
\end{equation}

By anonymity, $\p[\omega_2 > \omega_1 | q_1 \geq q_2] = \p[\omega_1 > \omega_2 | q_2 \geq q_1]$, therefore:
\begin{equation}
\p[\omega_1 > \omega_2 | q_1 \geq q_2] \geq \p[\omega_1 > \omega_2 | q_2 \geq q_1]
\end{equation}

\noindent from which the theorem follows by Bayes' theorem and the symmetric prior on $q_1, q_2$.
\end{proof}

\section{Discussion and Conclusion}
\label{conclusion}

We proved that market share indicates quality in the context of a model where customers base their strategy on their history with products and on market share data, under fairly weak restrictions on their behavior.

One consequence of the result is its guidance to the behavior of the customers themselves: A new customer, with no previous experience of the products, is advised to put her trust in market share data available. In a market in which customer-product interaction is one-shot, and all customers are equally informed about market share, all rational customers should behave alike.

The framework we used in deriving our results can be naturally generalized. While the restrictions of monotonicity and weak herding were successful for reaching the result, we do not claim that our formulation is the only one possible. Different formulations may be attempted, perhaps introducing other considerations into customer strategy. It should be apparent that generalizing our result would require definition of criteria similar to Criteria \ref{support} and \ref{support_e}, and constructing distributions whose support adhere to these new criteria.

For example, we believe that the requirement that all customers obey weak herding is too strict, and that some level of elitist customer behavior does not, in itself, invalidate the result. Namely, so long as elitism is outweighed (in some sense, to be defined) by herding behavior, inferences from market share to quality remain valid.

If elitism becomes the norm, this may not be true, as is illustrated by the following simple albeit artificial example:

\begin{example}
Let there be two products $1$ and $2$, with $1$ the superior product:  $q_1 > q_2$. Let there be $n$ customers, divided into two categories. Customer $1$'s partiality strategy to each product is $1$ if she has no prior history with the product or if her last experience with it was good, and $0$ otherwise. Customer $1$'s strategy ignores market share. Customers $2$ to $n$, on the other hand, are pure elitists: They will consume a product unconditionally unless that product is a leader in market share, in which case they will not consume it. The initial market share is $A = \{0,0\}$.

We analyze this to show that, on the 3rd round, market share does not indicate quality.

On the 1st round all customers consume all products. The market share after the round is $\omega = \{n,n\}$, so there is no market leader.

On the 2nd round, customers $2$ to $n$ will consume all products. Customer $1$, however, will do so only if her last round was a success. With probability $q_1(1-q_2)$ product $1$ will lead in market share, while with a smaller probability $q_2(1-q_1)$ (since $q_1 > q_2$) product $2$ will lead in market share. In other cases (probability $q_1q_2 + (1-q_1)(1-q_2)$), market share parity will continue.

On the 3rd round, if any product has larger market share, the elitists will stop consuming it, but will continue consuming the other product. The result will be that at the end of the round, market leadership will be reversed. As $1$ is the better product, and assuming $n$ large enough, the conclusion is that at the end of round $3$ higher market share indicates lower quality.
\end{example}

This example is artificial, {\em inter alia}, in that the negative result depends on the round number. In the example, market share leadership will oscillate, and even rounds will behave differently from odd rounds. We have not been able to find an economically realistic scenario that invalidates our central thesis.

Another generalization is by the introduction of money, which does not play a role in our current model due to the assumption of an undifferentiated market. This may be readily achieved by factoring price into quality, so that customer satisfaction is tied to perceived ``value for money''. It is not necessary for all customers to have the same sensitivity to price: The generalization (``market share indicates value for money'') clearly holds if all customers, were they fully informed about product qualities, agree that a certain product is preferable to another.

We were able to prove the case with competitively weak herding for two products only. Whether the result in fact holds for 3 products or more needs to be clarified. If the answer turns out to be negative, it will be interesting whether an alternative for competitively weak herding exists in which customers are responsive to market share ranking and for which the main result holds for any number of products.

Our analysis assumed that customers are aware of current market share only. What happens if we allow for awareness of historic market share values? There are two aspects to this: First, in the influence of market share data on customer strategies: We defined weak herding as a restriction on customers' response to {\em current} market share. If this definition of weak herding remains unchanged, then a weak-herding customer, even if aware of historic market share data, may only use it as a tie-breaker in forming her  strategies. It should be apparent that this does not disturb the validity of the results, as stated. Second, and however, the statement ``market share indicates quality", interpreted as a statement about the probability of better quality conditional on market share data {\em not} restricted to current data, may be invalid. This is especially true if customers are aware of trends. E.g. the leader in a market may be seen to be losing ground, while the follower is seen to be gaining ground.

We modeled quality as an unchanging attribute of a product. What happens if quality varies between rounds? The main result becomes moot. However it may be asked whether the intermediate results that show that market share is monotonically non-decreasing in the quality (Theorems \ref{theorem1} and \ref{theorem2}) remain true. Differentials in the quality may be replaced by partial differentials in any particular round's quality. The answer to this turns out to be negative in the general case. (For observe, e.g., that the pairing used in \eqref{h_ratio} means that product qualities from different periods are compared. This may cause the construction to fail by having negative joint probabilities. Interestingly, this problem does not exist, and so our results still hold, if product qualities are variable but non-increasing in time).

\end{document}